\documentclass[11pt]{article}
\usepackage{fullpage}
\usepackage{amsfonts,amsmath,amsthm,amssymb,mathtools,mathrsfs}
\usepackage{graphicx,upgreek,bm,xcolor}
\usepackage{booktabs}
\usepackage{algorithm}
\usepackage[noend]{algpseudocode}
\usepackage{float}
\usepackage{microtype}
\usepackage{enumitem}
\usepackage{aliascnt}
\usepackage{hyperref}
\usepackage{cleveref}
\hypersetup{hidelinks}
\newcommand{\Ber}{\mathrm{Ber}}

\makeatletter
\providecommand{\theHALG@line}{}
\renewcommand{\theHALG@line}{\thealgorithm.\arabic{ALG@line}}
\makeatother

\let\oldtheorem\newtheorem
\RenewDocumentCommand{\newtheorem}{s m o m O{}}{%
\IfBooleanTF{#1}%
{\oldtheorem{#2}{#4}}%
{\IfNoValueTF{#3}{\oldtheorem{#2}{#4}[#5]}%
{\newaliascnt{#2}{#3}%
\oldtheorem{#2}[#2]{#4}%
\aliascntresetthe{#2}}}}

\newcommand{\citet}{\cite}
\newcommand{\citep}{\cite}
\allowdisplaybreaks

\newtheorem{theorem}{Theorem}[section]
\newtheorem{lemma}[theorem]{Lemma}
\newtheorem{proposition}[theorem]{Proposition}
\newtheorem{corollary}[theorem]{Corollary}

\theoremstyle{definition}
\newtheorem{definition}[theorem]{Definition}
\newtheorem{remark}[theorem]{Remark}

\crefname{algorithm}{Algorithm}{Algorithms}
\crefname{section}{Section}{Sections}
\crefname{subsection}{Section}{Sections}

\newcommand{\R}{\mathbb{R}}
\newcommand{\E}{\mathbb{E}}

\newcommand{\one}{\mathbf{1}}

\newcommand{\cI}{\mathcal{I}}

\newcommand{\EA}{\mathsf{EA}}
\newcommand{\pos}[1]{\left(#1\right)^{+}}
\newcommand{\restr}[2]{#1\!\mid_{#2}}
\newcommand{\scr}{\mathcal}
\newcommand{\mb}{\mathbb}
\newcommand{\bP}{\mathbb{P}}
\newcommand{\Rev}{\mathsf{Rev}}
\newcommand{\Util}{\mathsf{Util}}

\newcommand{\I}{\mathcal{I}}
\newcommand{\B}{\mathcal{B}}
\newcommand{\cl}{\operatorname{cl}}
\newcommand{\conv}{\operatorname{conv}}
\newcommand{\argmax}{\operatorname*{argmax}}

\newcommand{\cM}{{\mathcal M}}
\newcommand{\cN}{{\mathcal N}}
\newcommand{\bx}{{\mathbf x}}
\newcommand{\bt}{{\mathbf t}}
\newcommand{\bv}{{\mathbf v}}
\newcommand{\btau}{\boldsymbol{\tau}}
\newcommand{\bp}{{\mathbf p}}
\newcommand{\by}{{\mathbf y}}
\newcommand{\bz}{{\mathbf z}}
\newcommand{\bu}{{\mathbf u}}

\newcommand{\bg}{{\mathbf g}}

\newcommand{\bdelta}{\boldsymbol{\delta}}
\newcommand{\brho}{\boldsymbol{\rho}}

\title{Prophet Inequalities and Online Contention Resolution for Matchoids}
\author{
Calum MacRury
\thanks{Georgia Tech, School of Industrial and Systems Engineering (ISyE), 
\texttt{calum.macrury@isye.gatech.edu}}
\and
Pranav Nuti
\thanks{University of Chicago, Booth School of Business, \texttt{pranavn@stanford.edu}}
\and
Jan Vondr\'ak
\thanks{Stanford University, Department of Mathematics, \texttt{jvondrak@stanford.edu}}
}
\date{}

\begin{document}
\maketitle

\begin{abstract}
In the classical prophet inequality, an algorithm observes a sequence of random variables with known distributions in an online fashion, and it must select one of the random variables with the goal of maximizing the expected value of its selection. The performance of the algorithm is compared to an \textit{omniscient prophet} who observes all of the random variables before having to make its selection. Combinatorial extensions of the classical prophet inequality in which the algorithm gets to pick a subset of the random variables (constrained to belong to some family of feasible sets) have been studied extensively.

We study prophet inequalities with a $k$-matchoid constraint (a common generalization of a $k$-matroid intersection constraint and a $k$-bounded hypergraph matching constraint) in two common online arrival models. We give guarantees with respect to the \textit{ex-ante} fractional relaxation of the omniscient prophet, obtaining an ex-ante competitive ratio of $\frac{1}{k+1}$ in the adversarial order case, and $\frac{1-e^{-k}}{k}$ in the random order case. Using the duality framework of Lee and Singla \cite{Lee2018}, this also yields online contention resolution schemes in these settings.

Our adversarial-order prophet inequality can be viewed as a generalization of a recent $\frac12$-competitive matroid prophet inequality by Kalantarzadeh and Pashkovich \cite{KP2026}. This generalization introduces a new framework: coordinated weighted principal partitions across multiple matroids. Our random-order prophet inequality is a generalization of the $k=1$ matroid case of Lee and Singla \cite{Lee2018}. The two results improve previously known competitive ratios for $k$-matroid intersection, which were $\frac{1}{(e+o(1))k}$ and $\frac{1}{k+1}$, respectively.
\end{abstract}

\newpage

\section{Introduction} \label{sec:introduction}

Prophet inequalities are a way to compare the performance of online and offline algorithms for various combinatorial optimization problems. Unlike worst-case online arrival models where the input sequence is chosen by an adversary, in prophet inequality problems, the online algorithm has pre-existing knowledge of the distribution on the elements of the input sequence. In recent years, there has been a significant increase in research on prophet inequalities, in part due to their applications in designing posted-price mechanisms in auctions \cite{ChawlaHMS10,DuttingFKL20,BHKKO24}, as well as to stochastic probing (see for example \cite{GuptaN13, GuptaNS16, GuptaNS17}).

The input to an algorithm for a prophet inequality problem is specified by a collection of elements $N$, whose weights $(W_i)_{i \in N}$ are drawn independently
from known distributions. Initially, the online algorithm does not know the instantiation of any of the weights. Instead, an adversary decides upon an ordering of the elements of $N$, and when element $i$ arrives, the value of $W_i$ is revealed to the algorithm. The algorithm must then irrevocably accept or reject $i$, where accepting $i$ gains a reward of $W_i$. In the classical ``single-item'' version of the problem, the online algorithm can select at most one element \cite{krengel1977semiamarts}.

More generally, we may consider a downward-closed set system on $N$ indicating which subsets of elements are feasible, i.e., simultaneously selectable. For instance, well-studied feasibility systems include independent sets of matroids, matchings, and knapsack constraints. The goal for the algorithm is to select a feasible subset $I \subseteq N$, for which the value of the (cumulative) reward, namely $\sum_{i \in I} W_i$, is as large as possible in expectation. The benchmark is the ``prophet'', who selects a feasible subset given full knowledge of the instantiations of the weights, and so maximizes the expected reward of their chosen subset. 

An online algorithm is $\alpha$-\textit{competitive}, provided its expected reward is an $\alpha$-fraction of the prophet's expected reward. The prophet \textit{secretary} problem is a related variant where the only difference is that the elements' weights are revealed in a uniformly random order, as opposed to an adversarial order. Designing an online algorithm is easier in this variant, as it gets to average over the arrival orders in its expected reward.

A common approach used to design a competitive online algorithm for a prophet inequality problem is to relax the prophet benchmark via a convex program, called the {\em ex-ante prophet benchmark}, and to try and obtain at least an $\alpha$-fraction of this benchmark (we formalize the program in \Cref{sec:ex-ante-benchmark}). This leads to a strictly stronger guarantee on the algorithm, and such an algorithm is often called $\alpha$-\textit{ex-ante competitive}. We call this guarantee an $\alpha$-ex-ante prophet inequality. 

One of the advantages of ex-ante prophet inequalities is that the hardest instances are more easily tractable. Specifically, it suffices to consider the case when each $W_i$ is a weighted Bernoulli random variable that is activated with probability $x_i$, in which case it takes a deterministic weight $v_i \ge 0$. Here $x_i$ should be thought of as the probability with which $i$ is selected by the prophet, and so we can also assume that the fractional solution $ \bx = (x_i)_{i \in N}$ is feasible in expectation (i.e., $\bx$ lies in the convex hull of all feasible subsets of $N$). In this special setting, the value of the convex program is $\sum_{i \in N} v_i x_i$, and so it suffices to design an online algorithm
whose expected reward is at least $\alpha \cdot \sum_{i \in N} v_i x_i$ to get an ex-ante prophet inequality.

A particularly effective way of doing this is via a randomized rounding tool known as an \textit{online contention resolution scheme (OCRS)}.
In the OCRS problem, each element $i$ of $N$ is \textit{active} independently with probability $x_i$, and
the active elements are revealed sequentially in some order. Upon learning that $i$ is active, the OCRS must irrevocably select/reject $i$, subject to feasibility constraints. The goal of the OCRS is to ensure each element $i$ is selected with probability $\alpha x_i$ for $\alpha$ as large as possible. In this case, it is called $\alpha$-\textit{selectable} on input $\bx$. By linearity of expectation,
an $\alpha$-selectable OCRS on input $\bx$ implies an online algorithm with expected reward $\alpha \sum_{i \in N} v_i x_i$, as formalized by \cite{feldman2015online}.
A similar reduction holds
for the prophet secretary problem, where one can instead use a \textit{random-order} contention resolution scheme (RCRS), which is an OCRS that processes the elements in random order.

In certain settings, explicit OCRSs with optimal selectability were not known, despite the existence of optimal prophet inequality algorithms (designed using other tools). For instance, for matroid constraints, \cite{kleinberg2012matroid} derived an optimal $1/2$-competitive algorithm via adaptive prices, despite the best explicit matroid OCRS selectability guarantee being $1/4$ \cite{feldman2015online}. This fact, together with the broader applicability of OCRSs as a rounding tool, motivated Lee and Singla \citet{Lee2018} to prove a converse to the reduction of \cite{feldman2015online}. Using LP duality, they showed that an $\alpha$-\textit{ex-ante} prophet inequality implies the existence of an $\alpha$-selectable OCRS. With this reduction, they applied the pricing-based approach of \cite{kleinberg2012matroid} to design an optimal $1/2$-selectable OCRS for matroids.

In this paper, we revisit the idea of designing contention resolution schemes using ex-ante prophet inequalities, but for more complicated constraint systems, namely intersections of $k$ matroids, and more generally $k$-matchoids (which generalize intersections of $k$ matroids as well as $k$-bounded hypergraph matching).

In the setting of $k$-matroid intersection,
\cite{kleinberg2012matroid} gave a $\frac{1}{4k-2}$-competitive prophet inequality, and
\cite{feldman2015online} gave a $\frac{1}{k+1}\left(1-\frac{1}{k+1}\right)^k = \frac{1}{(e+o(1))k}$-selectable
OCRS. In the random-order setting, \cite{adamczyk2018random} designed
a $\frac{1}{k+1}$-selectable RCRS for the intersection of $k$ matroids. On the other hand, for $k$-bounded hypergraph matchings, $1/(k+1)$ and $(1- e^{-k})/k$
are the ``baseline'' selectability guarantees for OCRSs and RCRSs, respectively, due to \cite{ezra2020online} and \cite{brubach2024offline}.\footnote{These papers focus on graph matchings (i.e., $k=2$), but the baseline arguments extend to $k$-bounded hypergraph matching; see \cite{MaMZ26Network} for details.} These latter results ought to ideally be true for $k$-matroid intersection, and more generally, for $k$-matchoids. We prove that this is indeed the case by providing ex-ante prophet inequalities with these respective ex-ante competitive ratios.

\subsection{Our Results}

A $k$-matchoid on a set of elements $N$ consists of matroids $\scr{M}_h=(N_h,\scr{I}_h)$ indexed by $h\in H$, where
$\bigcup_hN_h=N$ and each element $i \in N$ appears in at most $k$ of these matroids; that is, 
$
|\{h \in H:i\in N_h\}|\le k.
$
The family of feasible sets is
\[
    \scr{F}:=\{A\subseteq N:A\cap N_h \text{ is independent in } \scr{M}_h \text{ for every } h\in H\}.
\]
We denote the convex hull of the indicator vectors of the sets in $\scr{F}$ by $P(\scr{F})$.

Our first main result is an ex-ante prophet inequality for $k$-matchoids (and as a special case, for intersections of $k$ matroids):
\begin{theorem}\label{thm:ex-ante-adversarial}
There exists an efficient prophet inequality algorithm for $k$-matchoids which is $\frac{1}{k+1}$-ex-ante competitive. 
\end{theorem}

As mentioned in \Cref{sec:introduction}, in order to establish \Cref{thm:ex-ante-adversarial}, it suffices to consider the special case when each $W_i$ is distributed as $v_i \cdot \Ber(x_i)$ \footnote{Here $v_i \cdot \Ber(x_i)$ denotes a $\{0,v_i\}$-valued random variable, which takes value $v_i$ with probability $x_i$.}, where $v_i \ge 0$ is a non-negative weight,
and $\bx =(x_i)_{i \in N} \in P(\scr{F})$
(this reduction is implicitly proven in \cite{Lee2018}, as we expand upon in \Cref{sec:ex-ante-benchmark}).

Despite this greatly simplifying the inputs we must consider to prove \Cref{thm:ex-ante-adversarial}, the convex hull of the indicator vectors of the sets in $\scr{F}$ is difficult to work with. Instead, we work with the following natural relaxation:
\[
    \scr{P}_{\cap}
      :=\{\bx\in[0,1]^N:\restr{\bx}{N_h}\in P(\scr{M}_h)
        \text{ for every }h\},
\]
where $P(\scr{M}_h)$ is the convex hull of the indicator vectors of independent sets of $\scr{M}_h$. This relaxation is efficiently separable, and it is clear that $P(\scr{F}) \subseteq \scr{P}_{\cap}$ (typically with a strict containment). We establish our results by working with weighted Bernoulli random variables $W_i = v_i \cdot \Ber(x_i)$, and $\bx \in \scr{P}_{\cap}$, thereby obtaining results stronger than necessary for an ex-ante prophet inequality.

This lets us use the LP duality arguments of \cite{Lee2018} to obtain the following corollary (see \Cref{sec:extended_prelim} for a detailed derivation): 

\begin{corollary}
\label{cor:ocrs}
For every $k$-matchoid and every
$\bx\in\scr{P}_{\cap}$, there exists a $\frac{1}{k+1}$-selectable online contention resolution scheme (OCRS) on input $\bx$. Moreover, for each $\varepsilon > 0$, a $\left( \frac{1}{k+1} - \varepsilon \right)$-selectable OCRS can be designed efficiently.
\end{corollary}

This corollary improves the previously known selectability guarantee of $\frac{1}{(e+o(1)) k}$ for online contention resolution schemes for intersections of $k$ matroids. As we explain in \Cref{sec:random-element-ocrs}, the corollary also extends to a more general \textit{batched/random-element} OCRS setting where there is negative correlation amongst which elements are active. This batched setting was introduced in \cite{MaMZ26Network} for $k$-bounded hypergraph matchings to design randomized rounding based online algorithms for network revenue management and online combinatorial auctions. In this model, \cite{MaMZ26Network} proved $1/(1+k)$ is tight for $k$-bounded hypergraph matching when $k$ is a prime power, and so since our result applies to $k$-matchoids, it is tight for us as well. 

Our second contribution is an ex-ante prophet inequality for $k$-matchoids 
under random-order arrivals:

\begin{theorem} \label{thm:ex-ante-secretary}
    There exists an efficient prophet secretary algorithm for $k$-matchoids
    which is $\left(\frac{1-e^{-k}}{k}\right)$-ex-ante competitive.
\end{theorem}
Again, this implies an RCRS in the same setting:

\begin{corollary}
\label{cor:rcrs}
For every $k$-matchoid and every
$\bx\in\scr{P}_{\cap}$, there exists a $\left(\frac{1-e^{-k}}{k}\right)$-selectable random-order contention resolution scheme (RCRS) on input $\bx$. Moreover, for each $\varepsilon > 0$, a $\left( \frac{1-e^{-k}}{k}- \varepsilon \right)$-selectable RCRS can be designed efficiently.
\end{corollary}
This improves the previously known selectability guarantee of $\frac{1}{k+1}$ due to \cite{adamczyk2018random}. We also recover the known optimal result for $k=1$, and provide improvements
for $k \geq 2$. We note that a $(1-e^{-k})/k$-selectable RCRS was known for the special
case of the intersection of $k$ transversal matroids \cite{Adamczyk2025}, and so our result 
is a generalization of theirs. Finally, we mention that $\left(\frac{1-e^{-k}}{k}\right)$ \textit{is} beatable for $k$-bounded hypergraph matching \cite{MaMZ26Network},
so it remains an interesting open question whether $\left(\frac{1-e^{-k}}{k}\right)$ is tight
for $k$-matchoids or even intersections of $k$ matroids.

These selectability guarantees can also be compared with the known \textit{integrality gaps} of the polytope $\scr{P}_{\cap}$. An OCRS or RCRS for $\scr{P}_{\cap}$ cannot achieve a selectability guarantee better than the integrality gap. When $k =2$ or $k-1$ is a prime power, the known integrality gaps are $\frac{1}{k-1}$ for $k$-partite hypergraph matching, which is a special case of $k$-matroid intersection; and $\frac{1}{k-1+1/k}$ for $k$-bounded hypergraph matching, which is a special case of the $k$-matchoid problem (see \cite{chan2012linear} for these bounds). Therefore, our selectability guarantees are very close to the known integrality gaps.

\subsection{Our Techniques}

Suppose we have a ground set $N$, with weighted Bernoulli random variables,
$(W_i)_{i \in N}$, specified by deterministic weights $(v_i)_{i \in N}$ and 
activation probabilities $\bx=(x_i)_{i \in N}$, which lie in the polytope $\scr{P}_{\cap}$. For the purposes
of the overview, we assume that our constraint system is specified by the intersection
of $k$ matroids on $N$.

The prophet inequality for a single matroid of \cite{kleinberg2012matroid} is based on \textit{adaptive prices}. If $N = \{1, \ldots , |N|\}$, and the elements arrive in chronological order, suppose $I_{j-1}$ denotes the elements accepted right before $j$ is processed. Then, they compute the expected weight of the maximum independent set which is compatible with $I_{j-1}$, whose
value we denote $\tilde{R}(I_{j-1})$. If $I_{j-1} \cup \{j\}$ is independent,
and $W_j \ge \frac{1}{2}( \tilde{R}(I_{j-1}) - \tilde{R}(I_{j-1} \cup \{j\}))$, then they add $j$ to $I_{j-1}$
and proceed to the next element. Their analysis follows a revenue/utility style decomposition, where the crucial inequality they need to lower bound their utility is that if $I$ denotes the final set of accepted elements, and $I \cup B$ is independent for any set $B \subseteq N\setminus I$, then
\begin{equation} \label{eqn:key_inequality}
    \sum_{j \in B} (\tilde{R}(I_{j-1}) - \tilde{R}(I_{j-1}\cup \{j\})) \le \sum_{j \in B} (\tilde{R}(I) - \tilde{R}(I \cup \{j\}))\le \tilde{R}(I).
\end{equation}
While their online algorithm can naturally be generalized to an arbitrary constraint system on $N$, the first inequality of \eqref{eqn:key_inequality} relies on arguing that a certain function is submodular, which does not hold for intersections of $k$ matroids. For this generalized setting, they develop a $\frac{1}{4k-2}$-competitive prophet inequality using a modified pricing approach which cannot attain a guarantee of $\frac{1}{k+1}$. 

\paragraph{Our approach for prophet inequality.}
An inspiration for our prophet inequality is the recent work of \cite{KP2026},
which proposes a new approach to derive a $1/2$-competitive prophet inequality for a matroid constraint. (In addition, this prophet inequality is ``almost non-adaptive''; in some sense
fixing all its thresholds in advance.) The approach of \cite{KP2026} is based
on the notion of a {\em principal partition} of a matroid, which divides the ground
set into disjoint sets of decreasing ``density''. A certain weighted notion of 
a principal partition is used in \cite{KP2026} to set the value threshold for each
element in order to be selected by the algorithm.

We extend this idea to the setting of $k$ matroids, or more generally a $k$-matchoid,
which presents several difficulties. If we followed the analysis of \cite{KP2026}, each matroid would produce an unrelated principal partition of the ground set, so it is not clear how to coordinate these partitions into a unified analysis. The solution is a collection of weighted principal partitions
guided by a shared {\em surplus vector} $\bp$. The densities in each partition determine
a matroid-specific price vector $\bt_a$, and the thresholds are determined in turn
by the summation $\btau = \sum_a \bt_a$. However, this leads to a circular argument since the surplus vector $\bp$ is actually determined by the threshold vector $\btau$.

Resolving this circularity to find the surplus vector requires a fixed point which can be found by minimizing a certain convex potential function (and hence, it can even be found algorithmically in polynomial time). Identifying this fixed point, its important properties, and the appropriate potential function is the main technical contribution of this part of the paper. We remark that even in the case of $k=1$, our scheme is different from that of \cite{KP2026}.

We provide a detailed overview of our algorithm in \Cref{sec:detailed_overview_adversarial} prior to presenting the formal analysis in \Cref{sec:adversarial_detailed}.

\paragraph{Our approach for prophet secretary.}
Our starting point is the natural extension of the adaptive pricing algorithm of \cite{kleinberg2012matroid} to random-order arrivals, where we recall that $\tilde{R}(I)$ denotes
the expected weight of a feasible subset compatible with $I$. 

First, suppose that each $i \in N$ draws an arrival time $T_i \in [0,1]$ uniformly, and the elements are presented in increasing order
of arrival times. Then, for an appropriately chosen discount function $\gamma(t)$, if $I(t)$ denotes the elements accepted strictly before time $t$, element $i \in N$ with $T_i = t$ is accepted if possible, and 
\begin{equation} \label{eqn:continuous_acceptance}
W_i \ge \gamma(t)(\tilde{R}(I(t)) - \tilde{R}(I(t) \cup \{i\})).
\end{equation}
For a single matroid constraint, \cite{Ehsani2017} showed that for an appropriate choice of $\gamma$, this online algorithm is $(1-1/e)$-competitive. Afterwards, \cite{Lee2018} extended it to hold against the ex-ante benchmark.

When working with the generalization of this algorithm to the intersection of $k$ matroids,
in order for our algorithm to be efficient, first we must set the prices slightly differently.
Given a feasible subset $A \subseteq N$ and the linear objective $(v_j)_{j \in N}$, we compute an optimal residual vector $\by\le\bx$ to the restriction of $\scr{P}_{\cap}$, where $A$ is contracted. We denote $R(A)$ to be the value of this restricted LP.
Then, our algorithm is designed via adaptive prices as in \eqref{eqn:continuous_acceptance}, however we modify \eqref{eqn:continuous_acceptance} to be $W_i \ge \gamma_k(t)(R(I(t)) - R(I(t) \cup \{i\}))$, where $\gamma_k(t)$ is now a discount function dependent on $k$. 

In terms of the analysis, the main change is in lower bounding the online algorithm's utility. If
$\by =(y_i)_{i \in N}$ denotes the optimal residual vector to the restriction of $\scr{P}_{\cap}$ when 
$I(t)$ is contracted, 
then we argue that
\begin{equation} \label{eqn:local_exchange_fractional}
    \sum_{\text{$i\notin I(t): I(t) \cup \{i\}$ feasible}}  y_i\bigl(R(I(t))-R(I(t)\cup \{i\})\bigr)\le kR(I(t)).
\end{equation}
Roughly speaking, \eqref{eqn:local_exchange_fractional} can be seen as a fractional version of the second inequality in \eqref{eqn:key_inequality}, generalized to the intersection of $k$ matroids. Interestingly, for random-order arrivals, we are able to side-step the aforementioned submodularity step needed in the first inequality of \eqref{eqn:key_inequality}, which was the reason we abandoned adaptive prices for adversarial arrivals. The rest of the proof follows the standard utility--revenue calculation, except we are left with the differential equation
$1-k\gamma_k(t)+\gamma_k'(t)=0$, yielding $\gamma_k(0)=(1-e^{-k})/k$.

\section{Preliminaries}
\label{sec:preliminaries}

\subsection{Basic Notation}
\label{sec:notation}

We use $N$ for a ground set of elements.  For
$A\subseteq N$ and element $i \in N$, write $A+i:=A\cup\{i\}$ and $A-i:=A\setminus\{i\}$.
We bold all of our vectors $\bz$, and denote the $i$-th coordinate of $\bz$ by $z_i$.  Vector
inequalities are coordinatewise.  For $S\subseteq N$ and a vector $\bz$, let
\[
    z(S):=\sum_{i\in S}z_i,
    \quad
    \restr{\bz}{S}:=(z_i:i\in S).
\]
Let $\one_S$ be the incidence vector of $S$.  We write $\bu\cdot\bz$ for
the inner product, $\bu\odot\bz$ for coordinatewise multiplication, and
$\pos{z}:=\max\{z,0\}$.  Matroid labels precede coordinate labels in subscripts, such as in $t_{a,i}$.

\subsection{Matroids and Matchoids}
\label{sec:matroids-matchoids}

A matroid $\cM=(N,\I)$ consists of a nonempty family $\I\subseteq 2^N$ of
independent sets satisfying two axioms: every subset of an independent set is
independent, and whenever $I,J\in\I$ with $|I|<|J|$, there exists $j\in J\setminus I$
such that $I+j\in\I$.  A basis is an inclusion-wise maximal independent set.  We
write $\I(\cM)$ and $\B(\cM)$ for the independent sets and bases of $\cM$.
The rank of $S\subseteq N$ is
\[
    r_\cM(S):=\max\{|I|:I\subseteq S,\ I\in\I(\cM)\}.
\]
The independent-set and base polytopes are
\begin{align}
P(\cM)&:=\conv\{\one_I:I\in\I(\cM)\}=\{\bz\in\R^N_{\ge0}:z(S)\le r_\cM(S)
       \text{ for every }S\subseteq N\},\label{eq:ind-polytope}\\
B(\cM)&:=\conv\{\one_B:B\in\B(\cM)\}
=\{\bz\in P(\cM):z(N)=r_\cM(N)\}.\label{eq:base-polytope}
\end{align}
For $A\subseteq N$, the restriction $\cM|_A$ has ground set $A$ and
independent sets $\{I\in\I(\cM):I\subseteq A\}$.  The contraction $\cM/A$
has ground set $N\setminus A$ and rank function
\[
    r_{\cM/A}(S)=r_\cM(A\cup S)-r_\cM(A)
    \qquad(S\subseteq N\setminus A).
\]
The closure of $S$ is
\[
    \cl_\cM(S):=\{i\in N:r_\cM(S+i)=r_\cM(S)\}.
\]
If $\cN_1,\ldots,\cN_m$ are matroids on disjoint ground sets, their direct sum
$\cN_1\oplus\cdots\oplus\cN_m$ has independent sets
$I_1\cup\cdots\cup I_m$ with $I_j\in\I(\cN_j)$.  We write
$\cN\preceq\cM$ and call $\cN$ a \emph{strengthening} of $\cM$ when every
independent set of $\cN$ is independent in $\cM$.

A $k$-matchoid on a set of elements $N$ consists of matroids $\scr{M}_h=(N_h,\scr{I}_h)$ indexed by $h\in H$, where
$\bigcup_hN_h=N$ and
\[
    H(i):=\{h:i\in N_h\},
    \quad
    d(i):=|H(i)|\le k.
\]
We call these matroids the \emph{local matroids} of the matchoid, and the sets $N_h$ their
local ground sets. The family of feasible sets is
\[
    \scr F:=\{A\subseteq N:A\cap N_h\in\scr I_h
                  \text{ for every }h\in H\}.
\]
We call $A_h:=A\cap N_h$ the local restriction of $A$. Let
\[
    P(\scr{F}):=\conv\{\one_A:A\in\scr F\}.
\]
We work with the relaxation
\begin{equation}
    \scr P_\cap:=\bigl\{\bz\in[0,1]^N:
       \restr{\bz}{N_h}\in P(\scr M_h)\text{ for every }h\in H\bigr\}.
    \label{eq:local-relaxation}
\end{equation}
Every feasible set is locally independent, so $P(\scr{F})\subseteq\scr
P_\cap$.  The local relaxation has a polynomial-time separation oracle in
the standard independence-oracle model \cite{Edmonds1970,GLS1981}.

\subsection{The Ex-Ante Benchmark}
\label{sec:ex-ante-benchmark}

Let the values $(W_i)_{i\in N}$ of the prophet input be independent, nonnegative random variables
of finite expectation, with distributions $D_i$ and distribution functions
$F_i$.  Denote $F_i^{-1}$ as the generalized quantile function of $F_i$.  The
upper-tail reward curve of element $i$ is then
\[
    R_i(q):=\int_{1-q}^1 F_i^{-1}(u)\,du
    \qquad(q\in[0,1]).
\]
Here $R_i(q)$ is the expected contribution from retaining $W_i$
when it lies in its top $q$ quantile, with randomized tie-breaking at atoms.
For a downward-closed relaxation $\scr P\subseteq[0,1]^N$, define
\begin{equation}
    \EA_{\scr P}(\mathcal D):=\max_{\bx\in\scr P}\sum_{i\in N}R_i(x_i),
    \qquad \mathcal D:=(D_i)_{i\in N}.
    \label{eq:ex-ante-benchmark}
\end{equation}
The exact \textit{ex-ante benchmark} is defined with $\scr P=P(\scr{F})$. By interpreting $x_i$
as the probability the prophet selects element $i$, it is shown in \cite{feldman2015online} 
that the ex-ante benchmark upper bounds the prophet benchmark.

We in fact prove \Cref{thm:ex-ante-adversarial} and \Cref{thm:ex-ante-secretary} against
\eqref{eq:ex-ante-benchmark} with $\scr P=\scr P_\cap$.
Since $P(\scr{F}) \subseteq \scr P_\cap$, this only leads to a stronger result.
In what follows, we drop the dependence on $\scr{P}$ in $\EA_{\scr P}$,
as we always work with $\scr P=\scr P_\cap$.

For a weighted Bernoulli instance, i.e., $W_i=v_iX_i$, and $X_i\sim\Ber(x_i)$
for each $i \in N$, we say that $i$ is \textit{active} or that $i$ \textit{activates} if $X_i =1$.
Note that for such an input, the reward curve is $R_i(q)=v_i\min\{q,x_i\}$. If $\bx = (x_i)_{i \in N} \in \scr P_\cap$, then the
maximum in \eqref{eq:ex-ante-benchmark} is attained at $\bx$ and equals
\begin{equation}
    \EA(\bv,\bx):=\sum_{i\in N}v_ix_i.
    \label{eq:EA}
\end{equation}
For an arbitrary choice of $\bx\in\scr P_\cap$ and $\bv\geq0$,
we design online algorithms with expected reward at least $\alpha \cdot \EA(\bv,\bx)$, for the claimed choices of $\alpha$ in Theorems \ref{thm:ex-ante-adversarial} and \ref{thm:ex-ante-secretary}, respectively.
The general ex-ante prophet inequalities follow from
\Cref{lem:bernoulli-reduction} of \Cref{sec:lee-singla-reductions}, and the OCRS and RCRS statements follow from
\Cref{cor:weighted-reduction} of \Cref{sec:lee-singla-reductions}.

\section{Prophet Inequality: Outlining Our Approach} \label{sec:detailed_overview_adversarial}

In this section we focus on the intersection of $k$ matroids $\cM_1,\ldots,\cM_k$ on the common ground set $N$. This captures most of the technical difficulties; the full $k$-matchoid
argument is proved in \Cref{sec:random-element-ocrs}.

Fix a weighted Bernoulli instance $(\bv,\bx)$ as in
\Cref{sec:ex-ante-benchmark}.  We assume
\begin{equation}
\bx\in\bigcap_{a=1}^k P(\cM_a).\label{eq:bernoulli-feasible}
\end{equation}
If $x_i=0$, then $i$ is never active. We may delete such elements. After doing so, every remaining element is a nonloop in every $\cM_a$, since $0<x_i\le r_{\cM_a}(\{i\})$.

\paragraph{A high-level view.}

We seek a price vector $\bt_a = (t_{a,i} :i \in N)$ for each matroid, such that the thresholds $\tau_i$ to be used by the algorithm are defined as $\tau_i = \sum_{a=1}^{k} t_{a,i}$. We also define $p_i = (v_i - \tau_i) x_i$, the expected surplus of element $i$ relative to the threshold $\tau_i$. In vector form,
\[
\btau:=\sum_{a=1}^k \bt_a,
\qquad
\bp:=(\bv-\btau)\odot \bx.
\]
Let $I$ denote the output set. The intended meaning of the matroid-specific
prices $t_{a,i}$ is that they capture the lost profit of elements blocked due to
matroid $\cM_a$: More precisely, let $A$ denote the items of positive surplus
$p_i>0$ which are available for inclusion at the time they arrive.
The property we want is that $\tau(I) \geq p(N \setminus A)$.
If that's the case, then the expected value of the set $I$ found by the algorithm can be lower-bounded as 
$$ \E[v(I)] = \E[(\bv \odot \bx) \cdot \one_A] = \E[\bp \cdot \one_A + \btau \cdot \one_I]
\geq \E[\bp \cdot \one_A + \bp \cdot \one_{N \setminus A}] = \bp \cdot \one_N.$$
Let's denote $P = p(N) = \bp \cdot \one_N$; this is our lower bound on the performance of the algorithm.

On the other hand, we would like to claim that the ex-ante optimum is at most $(k+1) \cdot P$. This will be true if there is a relationship between the price vectors $\bt_a$ and the surplus vector $\bp$ such that 
the $\bt_a$-price of the ex-ante optimum is at most the global surplus: 
$\bt_a \cdot \bx \leq \bp \cdot \one_N = P$. Then, the ex-ante optimum is upper-bounded by
$$\EA(\bv,\bx) = \bv \cdot \bx = \btau \cdot \bx + \bp \cdot \one_N = \sum_{a=1}^{k} \bt_a \cdot \bx + \bp \cdot \one_N \leq (k+1) P.$$
The key to satisfying these bounds is finding the appropriate
vectors $\bt_a$ and $\bp$, which is the main technical hurdle in the proof.

\paragraph{Weighted principal partitions of a matroid.}
An important idea due to \cite{KP2026} is that the notion of a {\em principal partition} of a matroid can be useful for defining prices for a prophet inequality algorithm. An (unweighted) principal partition is obtained by identifying a subset of maximum ``density'' $|S|  / r_\cM(S)$, contracting this subset and recursing on the remaining matroid. The rough idea is that in a block of high density, we can only pick a small fraction of elements, and hence we should set a price proportional to density in order to select the most valuable elements. In order to prove a prophet inequality, we also have to take into account the elements' values $v_i$, and hence a certain variant of a {\em weighted principal partition} is used in \cite{KP2026}.

In our context, we consider a variation of this idea: a weighted principal partition determined by the surplus vector $\bp$ for each matroid $\cM_a$ as follows: Given $\bp$, we define a subset $U_{a,1} \in \argmax_{S \neq \emptyset} \frac{p(S)}{r_{\cM_a}(S)}$, contract it and recurse on the remaining matroid $\cM_a / U_{a,1}$, until no elements remain. This defines a partition $N = U_{a,1} \sqcup U_{a,2} \sqcup \ldots \sqcup U_{a,\ell}$; we also define $A_{a,0}=\varnothing$ and $A_{a,j} = U_{a,1} \cup \ldots \cup U_{a,j}$. (Keep in mind that we have a different partition for each matroid $\cM_a$.) 
Let's denote by $\cN_{a,j}$ the matroid $(\cM_a/A_{a,j-1})|_{U_{a,j}}$,
and define the matroid $\widehat{\cM}_a$ as the direct sum of the matroids $\cN_{a,j}$ for $j=1,2,\ldots,\ell$. In other words, for each $j$ we are allowed to take a subset of $U_{a,j}$ which does not violate independence in $\cN_{a,j}$. This is a more stringent requirement than being independent in $\cM_a$, and hence any selected set will be feasible as desired.

Each block of the partition is endowed with a per-element price which is exactly the respective density $\alpha_{a,j} = \frac{p(U_{a,j})}{r_{\cM_a / A_{a,j-1}}(U_{a,j})}$. It follows from the construction that these prices are decreasing: $\alpha_{a,1} > \alpha_{a,2} > \ldots > \alpha_{a,\ell}$. 
We set $t_{a,i} = \alpha_{a,j}$ for each element $i \in U_{a,j}$.
Adding up over the $k$ matroids, this defines a total price $\tau_i = \sum_{a=1}^{k} t_{a,i}$,
which is the threshold the algorithm uses for selection of $i$.

However, observe that this construction has a circular structure: 
The surplus vector $\bp$ determines the principal partition for each matroid, which determines the price vectors $\bt_a$, and $\btau = \sum_{a=1}^{k} \bt_a$ in turn determines the surplus vector $\bp$. Resolving this circular argument requires finding a fixed point satisfying some desired properties. We formalize these properties as follows:
\begin{enumerate}
\item [(1)] $\bp  = (\bv - \sum_{a=1}^{k} \bt_a) \odot \bx \geq 0$.
\item [(2)] $\bp = \bt_a \odot \by_a$ for each $a=1,\ldots,k$.
\item [(3)] $\by_a \in \argmax_{\by \in B(\cM_a)} \bt_a \cdot \by$ for each $a=1,\ldots,k$.
\end{enumerate}
It can be verified that these conditions imply that the partition defined by the level sets
of $\bt_a$ for each matroid $\cM_a$ is a $\bp$-weighted principal partition. It remains to find
the desired vectors satisfying (1), (2), (3).

\paragraph{Finding the price and surplus vectors.}
As suggested previously, the desired vectors can be written as a fixed point of a certain map,
whose existence follows from Kakutani’s fixed-point theorem. However, since we also aim to construct these
vectors efficiently, we instead introduce a convex potential function whose minimizer yields the
required vectors.

For each matroid $\cM_a$, define
\[
\Gamma_{\cM_a}(\bt_a):= \frac12 \max_{B\in\B(\cM_a)} \sum_{i\in B} (t_{a,i})^2.
\]
We note that the subgradients of $\Gamma_{\cM_a}$ are vectors of the form $\bt_a \odot \by_a$ such that $\by_a$ is $\bt_a$-maximizing among fractional bases in $\cM_a$; i.e. the candidate vectors for condition (2).
We set $\bt_1,\ldots,\bt_k$ to be vectors minimizing the following convex potential function:
\begin{eqnarray}
\Phi(\bt_1,\ldots,\bt_k) & = & \sum_{a=1}^k\Gamma_{\cM_a}(\bt_a) + \frac12\sum_{i\in N}x_i(v_i-\tau_i)^2 \\
& = & \frac12 \sum_{a=1}^{k} \max_{B\in\B(\cM_a)} \sum_{i\in B} (t_{a,i})^2  +\frac12\sum_{i\in N} \left(v_i-\sum_{a=1}^k t_{a,i}\right)^2 x_i
\, \, \text{, over }\bt_1,\ldots,\bt_k\ge 0.\label{eq:program-idea}
\end{eqnarray}
The optimality conditions imply condition (1), the surplus vector $\bp = (\bv - \btau) \odot \bx$ is nonnegative, and for every matroid $\cM_a$, $\bp$ is a subgradient of $\Gamma_{\cM_a}$ at $\bt_a$. This is exactly condition (2), $\bp = \bt_a \odot \by_a$, simultaneously for each matroid $\cM_a$. If we define $U_{a,j}$ as the sets of elements attaining distinct values of $t_{a,i}$, it can be shown that this is exactly the $\bp$-weighted principal partition for each matroid $\cM_a$. 
We obtain all the necessary properties for the analysis from here.

\paragraph{Lower-bounding the algorithm.}
Consider now a particular block $U_{a,j}$ and the associated matroid $\cN_{a,j}$. 
From the properties of a weighted principal partition, we obtain the following for any $S \subseteq U_{a,j}$:
\begin{equation}
p(S) \le \alpha_{a,j} r_{\cN_{a,j}}(S),
\qquad
p(U_{a,j}) = \alpha_{a,j} r_{\cN_{a,j}}(U_{a,j}).\label{eq:density-idea}
\end{equation}
In other words, $U_{a,j}$ has the maximum density among all subsets of $U_{a,j}$, and the maximizing density is exactly $\alpha_{a,j}$. 

Suppose that our algorithm selects a set $I$, and $C_{a,j}$ is the subset of the block $U_{a,j}$ blocked by $I \cap U_{a,j}$ in the matroid $\cN_{a,j}$, i.e. $C_{a,j} = \cl_{\cN_{a,j}}(I \cap U_{a,j})$. By the definition of the closure, $r_{\cN_{a,j}}(C_{a,j}) = r_{\cN_{a,j}}(I \cap U_{a,j}) = |I \cap U_{a,j}|$. Hence we have
\[
p(C_{a,j}) \le \alpha_{a,j} r_{\cN_{a,j}}(C_{a,j}) = \alpha_{a,j} |I\cap U_{a,j}| = t_a(I \cap U_{a,j}).
\]
The last expression is exactly the $\bt_a$-price of the accepted elements in this block. Each blocked
element in $N$ can be charged to some matroid in which it is blocked, and so the surplus of all the blocked elements $C$ can be bounded as
\begin{equation}
p(C)\le \sum_{a=1}^k \sum_j t_a(I \cap U_{a,j}) = \sum_{a=1}^{k} t_a(I) = \tau(I).\label{eq:blocking-idea}
\end{equation}
Thus the surplus lost through blocking is paid for by the price of the selected set.

Now we can lower-bound $\E[v(I)]$: Let $A$ denote the set of elements of positive surplus $p_i>0$, available for inclusion when they arrive. If $i \in A$ is also active, it will be selected. Hence
$$ \E[v(I)] = \E[\sum_{i \in A} v_i x_i] = \E[\sum_{i \in A} (p_i + \tau_i x_i)]
\geq \E[\sum_{i \in N \setminus C} p_i + \sum_{i \in I} \tau_i]
= \E[p(N) - p(C) + \tau(I)] \geq p(N) = P.$$

\paragraph{Upper-bounding the ex-ante optimum.}
Finally, consider the ex-ante optimum $\bx$. Since $\bx \in P(\cM_a)$ for each $\cM_a$, we have
$\bt_a \cdot \bx = \sum_{j=1}^{\ell} \alpha_{a,j} x(U_{a,j})$ and we have $\alpha_{a,1} > \alpha_{a,2} > \ldots$. Hence, this expression is maximized over $\by \in P(\cM_a)$ by the greedy algorithm, assigning as much mass as possible successively to $y(U_{a,1})$, $y(U_{a,2})$, etc. The optimal choice is $y(U_{a,j}) = r_{\cM_a}(A_{a,j}) - r_{\cM_a}(A_{a,j-1}) = r_{\cN_{a,j}}(U_{a,j})$, which gives 
$$ \bt_a \cdot \bx \leq \bt_a \cdot \by = \sum_{j=1}^{\ell} \alpha_{a,j} r_{\cN_{a,j}}(U_{a,j}) = \sum_{j=1}^{\ell} p(U_{a,j}) = P.$$
We conclude that
$$ \EA(\bv,\bx) = \bv \cdot \bx = \bp \cdot \one_N + \sum_{a=1}^{k} \bt_a \cdot \bx \leq (k+1) P, $$
which gives the desired ex-ante competitive ratio.

In the following, we give a formal proof based on this outline.

\section{Prophet Inequality: Detailed Proofs} \label{sec:adversarial_detailed}

For a matroid $\cM$ on $N$, define
\begin{equation}
\Gamma_\cM(\bt):=\frac12\max_{B\in\B(\cM)}\sum_{i\in B} t_i^2
=\frac12\max_{\by\in B(\cM)} (\bt\odot\bt) \cdot \by.\label{eq:Gamma}
\end{equation}
The next lemma contains all the matroid facts used in the proof.
\begin{lemma}[Quadratic prices and principal blocks]\label{lem:principal-blocks}
For every $\bt\in\R^N$,
\begin{equation}
\partial\Gamma_\cM(\bt)
=\left\{\bt\odot \by:\by\in B(\cM),\ \frac12 (\bt\odot\bt) \cdot \by = \Gamma_\cM(\bt)\right\}.\label{eq:subdiff}
\end{equation}
Suppose now that $\bt\ge0$. Let $\alpha_1>\cdots>\alpha_\ell\ge0$ be the distinct coordinates of $\bt$,
and put
\begin{equation}
A_0:=\varnothing,
\qquad
A_j:=\{i:t_i\ge\alpha_j\},
\qquad
U_j:=A_j\setminus A_{j-1}.\label{eq:levels}
\end{equation}
Define
\begin{equation}
\cN_j:=(\cM/A_{j-1})|_{U_j},
\qquad
\widehat \cM(\bt):=\cN_1\oplus\cdots\oplus\cN_\ell.\label{eq:direct-sum}
\end{equation}
Then
\begin{equation}
B(\widehat \cM(\bt))
=\argmax_{\by\in B(\cM)} \bt \cdot \by
=\argmax_{\by\in B(\cM)}(\bt\odot\bt) \cdot \by,\label{eq:face}
\end{equation}
and $\widehat \cM(\bt)\preceq \cM$.

If $\bp\in\partial\Gamma_\cM(\bt)$, then for every $j$ and every $S\subseteq N\setminus A_{j-1}$,
\begin{equation}
p(S)\le\alpha_j r_{\cM/A_{j-1}}(S),
\qquad
p(U_j)=\alpha_j r_{\cN_j}(U_j).\label{eq:residual-density}
\end{equation}
In particular,
\begin{equation}
\restr{\bp}{U_j}\in\alpha_j B(\cN_j),
\qquad
p(S)\le\alpha_j r_{\cN_j}(S)\quad(S\subseteq U_j).\label{eq:block-density}
\end{equation}
Finally, for every $\bx \in P(\cM)$,
\begin{equation}
\bt \cdot \bx \le p(N).\label{eq:price-budget}
\end{equation}
\end{lemma}

\begin{proof}
For a base $B$, let
\[
\gamma_B(\bt):=\frac12\sum_{i\in B}t_i^2.
\]
Then $\Gamma_\cM=\max_{B\in\B(\cM)}\gamma_B$ and
$\nabla\gamma_B(\bt)=\bt\odot\one_B$. The subdifferential of a finite maximum of functions is the convex hull of the gradients of the maximizing functions. This gives \eqref{eq:subdiff}.

Suppose $\bt\ge0$ and consider $A_1, A_2,\ldots$ defined as above. By the greedy algorithm for matroids, a base $B$ has maximum $\bt$-weight $t(B) = \bt \cdot \one_B$ if and only if
\begin{equation}
|B\cap A_j|=r_\cM(A_j)\qquad(j=1,\ldots,\ell).\label{eq:tight-levels}
\end{equation}
The same bases maximize $\bt\odot\bt$, since $\bt$ and $\bt\odot\bt$ induce the same ordering of the elements. The
equalities in \eqref{eq:tight-levels} hold if and only if $B\cap U_j$ is a base of
$(\cM/A_{j-1})|_{U_j}$ for every $j$. Indeed, if \eqref{eq:tight-levels} holds, then
$B\cap A_{j-1}$ is a base of $\cM|_{A_{j-1}}$ and $B\cap A_j$ is independent. Thus $B\cap U_j$ is
independent in $\cM/A_{j-1}$, and its cardinality is
$r_\cM(A_j)-r_\cM(A_{j-1})=r_{\cN_j}(U_j)$, so it is a base of $\cN_j$. Conversely, if $B\cap U_j$ is a
base of $\cN_j$ for every $j$, then summing the block cardinalities gives
$|B\cap A_j|=r_\cM(A_j)$ for every $j$. Taking convex hulls gives \eqref{eq:face}.

Every base of $\widehat \cM(\bt)$ is therefore a base of $\cM$. Since every independent set of
$\widehat \cM(\bt)$ extends to one of its bases, it is also independent in $\cM$. Thus
$\widehat \cM(\bt)\preceq \cM$.

By \eqref{eq:subdiff} and \eqref{eq:face}, we may write $\bp=\bt\odot \by$ for some
$\by\in B(\widehat \cM(\bt))$. In particular, $y(A_j)=r_\cM(A_j)$ for every $j$. Fix $j$ and
$S\subseteq N\setminus A_{j-1}$. Since $t_i\le\alpha_j$ outside $A_{j-1}$,
\[
p(S)\le\alpha_j y(S)
\le\alpha_j\bigl(r_\cM(A_{j-1}\cup S)-r_\cM(A_{j-1})\bigr).
\]
Taking $S=U_j$, both inequalities are equalities. This proves \eqref{eq:residual-density}. Since
$\restr{\by}{U_j}\in B(\cN_j)$ and
$\restr{\bp}{U_j}=\alpha_j\restr{\by}{U_j}$, it also proves
\eqref{eq:block-density}. Finally, since $\bt\ge0$, a maximum-$\bt$ independent set may be extended to
a base without losing weight. Therefore, for $\bx\in P(\cM)$,
\[
\bx \cdot \bt \le \max_{\bz \in B(\cM)} \bt \cdot \bz = \bt \cdot \by = p(N).
\]
\end{proof}

Equation \eqref{eq:residual-density} is the density description of a weighted principal partition of
$(\cM,\bp)$. After $A_{j-1}$ is contracted, the block $U_j$ has maximum density $\alpha_j$. We do not
impose a maximality convention when several densest sets tie; rank-zero blocks may be discarded.
With this convention, the level sets of $\bt$ recover a $\bp$-weighted principal partition for $\cM$.

We now choose the prices for all $k$ matroids at once.

\begin{proposition}[Common-surplus equilibrium]\label{prop:equilibrium}
Suppose $x_i>0$ for every $i$ and \eqref{eq:bernoulli-feasible} holds. There are vectors
$\bt_1,\ldots,\bt_k\in\R^N_{\ge0}$ and $\bp\in\R^N_{\ge0}$ such that, with
$\btau:=\sum_{a=1}^k \bt_a$,
\begin{equation}
\bp=\bx\odot(\bv-\btau)\label{eq:surplus}
\end{equation}
and
\begin{equation}
\bp\in\partial\Gamma_{\cM_a}(\bt_a)\qquad(a=1,\ldots,k).\label{eq:common-subgradient}
\end{equation}
In particular, $\btau\le \bv$ and
\begin{equation}
x_iv_i=p_i+x_i\tau_i\qquad(i\in N).\label{eq:exact-decomposition}
\end{equation}
\end{proposition}

\begin{proof}
Consider the convex function
\begin{equation}
\Phi(\bt_1,\ldots,\bt_k)
:=\sum_{a=1}^k\Gamma_{\cM_a}(\bt_a)
+\frac12\sum_{i\in N}x_i\left(v_i-\sum_{a=1}^k t_{a,i}\right)^2\label{eq:joint-program}
\end{equation}
over $(\R^N_{\ge0})^k$. Every $\cM_a$ is loopless. Hence every element belongs to a base, and
\[
\Gamma_{\cM_a}(\bt_a)\ge\frac12\|\bt_a\|_\infty^2.
\]
Thus $\Phi$ is coercive and has a minimizer.

Fix a minimizer, let $\btau:=\sum_a \bt_a$, and put $\bp:=\bx\odot(\bv-\btau)$. The constrained first-order
condition in block $\bt_a$ is
\[
0\in\partial\Gamma_{\cM_a}(\bt_a)-\bp+N_{\R^N_{\ge0}}(\bt_a),
\]
where $N_{\R^N_{\ge0}}(\bt_a)$ denotes the normal cone. A vector in this normal cone has coordinate
zero when $t_{a,i}>0$ and a nonpositive coordinate when $t_{a,i}=0$. Hence, for every $a$, there is a
vector $\bg_a\in\partial\Gamma_{\cM_a}(\bt_a)$ such that
\begin{equation}
g_{a,i}=p_i\quad\text{if }t_{a,i}>0,
\qquad
g_{a,i}\ge p_i\quad\text{if }t_{a,i}=0.\label{eq:boundary-optimality}
\end{equation}
By \eqref{eq:subdiff}, we may write $\bg_a=\bt_a\odot \by_a$ for some $\by_a\in B(\cM_a)$. In particular,
$g_{a,i}\ge0$, and $g_{a,i}=0$ whenever $t_{a,i}=0$.

Suppose $p_i<0$. The first condition in \eqref{eq:boundary-optimality} then rules out $t_{a,i}>0$
for every $a$. Thus $\tau_i=0$, and $p_i=x_iv_i\ge0$, a contradiction. Hence $\bp\ge0$. If
$t_{a,i}=0$, the second condition in \eqref{eq:boundary-optimality} now gives
$0=g_{a,i}\ge p_i\ge0$. Thus $g_{a,i}=p_i$ also at the boundary. We conclude that $\bg_a=\bp$ for every
$a$, which proves \eqref{eq:common-subgradient}. Finally, \eqref{eq:surplus} and $\bp\ge0$ give
$\btau\le \bv$, and \eqref{eq:exact-decomposition} follows.
\end{proof}

\subsection{The algorithm}
\textbf{Preprocessing.} Given $\bx\in\bigcap_{a=1}^kP(\cM_a)$ and $\bv\ge0$, let
$N_+:=\{i\in N:x_i>0\}$ and $\cM_a^+:=\cM_a|_{N_+}$. Apply
\Cref{prop:equilibrium} to the restricted instance to obtain
$\bt_1,\ldots,\bt_k$ and $\bp$, and set
$\tau_i:=\sum_{a=1}^k t_{a,i}$ for $i\in N_+$. For every $a$, let
$\widehat\cM_a^+$ be the direct-sum matroid from \eqref{eq:direct-sum} for
$(\cM_a^+,\bt_a)$. Extend it to a matroid $\widehat\cM_a$ on $N$ by making
every element of $N\setminus N_+$ a loop, and set $\tau_i:=v_i$ for those elements.
These parameters are fixed before any activations are revealed.

\begin{algorithm}[H]
\caption{Coordinated-price adversarial-order prophet inequality algorithm}
\label{alg:adversarial-prophet}
\begin{algorithmic}[1]
\Require Matroids $\cM_1,\ldots,\cM_k$, activation vector $\bx$, and
         values $\bv\ge0$
\Ensure A feasible accepted set $I$
\State Construct $\btau$ and $(\widehat\cM_a)_{a=1}^k$ by the preprocessing above
\State $I\gets\varnothing$
\For{each element $i$ in the fixed arrival order}
    \State Reveal $X_i$
    \If{$X_i=1$, $v_i>\tau_i$, and
         $I+i\in\bigcap_{a=1}^k\I(\widehat\cM_a)$}
        \State $I\gets I+i$
    \EndIf
\EndFor
\State \Return $I$
\end{algorithmic}
\end{algorithm}

\begin{theorem}[Adversarial-order ex-ante prophet inequality]
\label{thm:adversarial-prophet}
Let $\bx\in\bigcap_{a=1}^kP(\cM_a)$ and $\bv\in\R^N_{\ge0}$, and let
$X_i\sim\Ber(x_i)$ independently. For every fixed arrival order,
\Cref{alg:adversarial-prophet} returns a set
$I\in\bigcap_{a=1}^k\I(\cM_a)$ satisfying
\[
    \E[v(I)]\ge\frac{1}{k+1}\sum_{i\in N}x_iv_i.
\]
\end{theorem}

\begin{proof}
Work on $N_+$ and suppress the superscript ${}^+$. Put
\[
P:=p(N),
\qquad
Q:=\{i:p_i>0\}.
\]
By \eqref{eq:surplus}, $Q=\{i:v_i>\tau_i\}$. More, \Cref{lem:principal-blocks} gives
$\widehat \cM_a\preceq \cM_a$ on the reduced ground set.

By construction, \Cref{alg:adversarial-prophet} accepts an element precisely when it is active,
belongs to $Q$, and is feasible in every $\widehat \cM_a$. It is therefore feasible in every original matroid. We next prove the estimate
which pays for blocked surplus.

For an element $i\in Q$, call $i$ available if its addition to the set accepted before its arrival is
independent in every $\widehat \cM_a$. Availability is defined without regard to whether $i$ is active.
Let $C$ be the set of elements in $Q$ which are unavailable when they arrive. We claim that, for
every realization,
\begin{equation}
p(C)\le\tau(I).\label{eq:blocked-surplus}
\end{equation}
For every $i\in C$, choose one matroid $a$ in which it is blocked. Let $U_{a,j}$ be the price block
containing $i$, let $\cN_{a,j}=(\cM_a/A_{a,j-1})|_{U_{a,j}}$, and assign $i$ to the pair $(a,j)$. Denote the
assigned set by $C_{a,j}$. If $I_{<i}$ is the set accepted before $i$ arrives, then
\[
i\in\cl_{\cN_{a,j}}(I_{<i}\cap U_{a,j}).
\]
Closure is monotone, and hence
\[
C_{a,j}\subseteq\cl_{\cN_{a,j}}(I\cap U_{a,j}).
\]
Since $I\cap U_{a,j}$ is independent in $\cN_{a,j}$,
\[
r_{\cN_{a,j}}(C_{a,j})\le|I\cap U_{a,j}|.
\]
The local density inequality \eqref{eq:block-density} now gives
\[
p(C_{a,j})\le\alpha_{a,j} r_{\cN_{a,j}}(C_{a,j})
\le\alpha_{a,j}|I\cap U_{a,j}|.
\]
Summing over all pairs $(a,j)$ proves
\[
p(C)\le\sum_{a=1}^k\sum_j\alpha_{a,j}|I\cap U_{a,j}|
=\sum_{a=1}^k t_a(I)=\tau(I),
\]
which is \eqref{eq:blocked-surplus}.

For $i\in Q$, let $A_i$ be the event that $i$ is available. It depends only
on activity indicators revealed before the arrival of $i$, and is independent
of $X_i$. Thus
\[
\bP(i\in I)=x_i\bP(A_i).
\]
Using \eqref{eq:exact-decomposition},
\begin{align*}
\E[v(I)]
&=\sum_{i\in Q}(p_i+x_i\tau_i)\bP(A_i)\\
&=P-\E[p(C)]+\E[\tau(I)] \ge P,
\end{align*}
where the last inequality follows from \eqref{eq:blocked-surplus}.
It remains to compare $P$ with the ex-ante value. For every $a$, $\bx\in P(\cM_a)$, so
\eqref{eq:price-budget} gives
\[
\bt_a \cdot \bx \le p(N)=P.
\]
Summing these inequalities and using \eqref{eq:exact-decomposition},
\[
\sum_{i\in N} x_i v_i = P+\sum_{a=1}^k \bx\cdot \bt_a\le(k+1)P.
\]
Together with the preceding inequality, this proves the guarantee on the
support of $\bx$. The remaining elements are loops in every
$\widehat\cM_a$ and have $x_i=0$, so they are never accepted and contribute
nothing to the benchmark.
\end{proof}

\section[Prophet Secretary]{Prophet Secretary}
\label{sec:algorithm}

Fix $\bx\in\scr P_\cap$ and nonnegative values $\bv$.  For a feasible set
$A\in\scr F$, let $\scr M_h/A_h$ denote contraction by the local restriction
$A_h$ defined in \Cref{sec:matroids-matchoids}.
Define the residual relaxation
\begin{equation}
\begin{split}
    \scr{P}(A):=\bigl\{\by\in\mb{R}_{\ge0}^N:
       &\ y_i=0\text{ for }i\in A,\\
       &\ \restr{\by}{N_h\setminus A_h}
          \in P(\scr{M}_h/A_h)
          \text{ for every }h\in H\bigr\}.
\end{split}
    \label{eq:residual-relaxation}
\end{equation}
The relaxation is downward closed.  Moreover,
$\bx\in\scr{P}(\varnothing)$ by \eqref{eq:local-relaxation}.
Define the residual value
\begin{equation}
    R(A):=\max\bigl\{\bv\cdot\by:\by\in\scr{P}(A),\ 0\le\by\le\bx\bigr\}.
    \label{eq:residual-value}
\end{equation}
The bound $\by\le\bx$ is coordinatewise.  If $i\notin A$ and
$A+i=A\cup\{i\}$ is feasible, define the residual drop
\begin{equation}
    \Delta_i(A):=R(A)-R(A+i)\ge0.
    \label{eq:residual-drop}
\end{equation}
The inequality follows from contraction and that the relaxation is downward closed.
Let
\begin{equation}
    \gamma(t):=\frac{1-e^{-k(1-t)}}{k},
    \qquad t\in[0,1].
    \label{eq:gamma}
\end{equation}
This function satisfies
\begin{equation}
    \gamma(1)=0,
    \qquad
    1-k\gamma(t)+\gamma'(t)=0,
    \qquad
    \gamma(0)=\frac{1-e^{-k}}{k}.
    \label{eq:gamma-ode}
\end{equation}

\begin{algorithm}[H]
\caption{Residual-price prophet secretary algorithm}
\label{alg:prophet}
\begin{algorithmic}[1]
\Require Local matroids $(\scr{M}_h)_{h\in H}$,
         activation vector $\bx$, and values $\bv\ge0$
\Ensure A feasible accepted set $I$
\State Draw independent clocks $T_i\sim\operatorname{Unif}[0,1]$ and process
       elements in increasing clock order
\State $I\gets\varnothing$
\For{each element $i$ arriving at time $t=T_i$}
    \State Reveal $X_i$
    \If{$X_i=1$ and $I+i\in\scr{F}$}
        \State Compute $\Delta_i(I)=R(I)-R(I+i)$ using
               \eqref{eq:residual-value}
        \If{$v_i\ge\gamma(t)\Delta_i(I)$}
            \State $I\gets I+i$
        \EndIf
    \EndIf
\EndFor
\State \Return $I$
\end{algorithmic}
\end{algorithm}

We use the random-clock representation from \Cref{sec:ocrs}.  If the input is
presented as a permutation, assign the sorted values of $|N|$ independent
uniform clocks to its successive positions.

\subsection{An Exchange Lemma for Matchoids}
\label{sec:structural}

We now prove the structural inequality that replaces the single-matroid
critical-value bound in the Lee--Singla analysis. We start with the following standard lemma:

\begin{lemma}[Exchange map]
\label{lem:exchange}
Let $\scr{M}=(N,\scr{I})$ be a matroid and let $I,J\in\scr{I}$.
There is a map $\sigma_{I,J}:J\to I\cup\{\bot\}$ such that
\begin{enumerate}[label=(\roman*),leftmargin=2em]
    \item $(I\setminus\{\sigma_{I,J}(i)\})\cup\{i\}\in\scr{I}$ for every
          $i\in J$, where removing $\bot$ has no effect;
    \item each element of $I$ has at most one preimage; and
    \item $\sigma_{I,J}(i)=i$ whenever $i\in I\cap J$.
\end{enumerate}
\end{lemma}

\begin{proof}
Extend $I$ and $J$ to bases $B_I$ and $B_J$ of $\scr{M}$.  The bijective
basis-exchange theorem gives a bijection $\phi:B_J\to B_I$ that fixes
$B_I\cap B_J$ and for which $B_I-\phi(i)+i$ is a basis for every
$i\in B_J$; see \citet{Brualdi1969}.  For $i\in J$, set
\[
    \sigma_{I,J}(i)
    :=\begin{cases}
        \phi(i),&\phi(i)\in I,\\
        \bot,&\phi(i)\notin I.
      \end{cases}
\]
Injectivity away from $\bot$ follows from that of $\phi$, and every element of $I\cap J$ is
fixed.  If $\phi(i)\in I$, then $I-\phi(i)+i$ is a subset of the exchanged
basis.  If $\phi(i)\notin I$, then $I+i$ is a subset of that basis.  Both sets
are independent.
\end{proof}

The identity property in \Cref{lem:exchange} is important.  When $i\in I$,
the exchanged residual set must remove $i$ before $i$ is contracted; otherwise
the purported vector in the contracted polytope could retain a positive
$i$-coordinate.

\begin{lemma}[Matchoid charging]
\label{lem:charging}
Let $A\in\scr{F}$, and let $\by$ be an optimal solution of
\eqref{eq:residual-value}.  Then
\begin{equation}
    \sum_{i\notin A:\,A+i\in\scr{F}}y_i\Delta_i(A)
       \le kR(A).
    \label{eq:charging}
\end{equation}
\end{lemma}

\begin{proof}
Fix a local component $h\in H$.  Write
\[
    \scr{N}_h:=\scr{M}_h/A_h,
    \qquad
    \bar{\by}_h:=\restr{\by}{N_h\setminus A_h}.
\]
By feasibility, $\bar{\by}_h\in P(\scr{N}_h)$.  The independent-set
polytope of a matroid is integral, so choose a distribution $\nu_h$ over
independent sets of $\scr{N}_h$ whose marginal vector is
$\bar{\by}_h$.  Let $I$ and $J$ be independent samples from $\nu_h$.

For $J$ in the support of $\nu_h$ and $i\in J$, apply
\Cref{lem:exchange} to $I,J$ and define the vector
\begin{equation}
    \bdelta_h(J,i)
      :=\E_{I\sim\nu_h}
          \bigl[\one_{\{\sigma_{I,J}(i)\}}\bigr],
    \label{eq:delta}
\end{equation}
where $\one_{\{\bot\}}$ is the zero vector.  Injectivity of the exchange map away from $\bot$
gives, coordinatewise,
\begin{equation}
    \sum_{i\in J}\bdelta_h(J,i)
       \le \E_I[\one_I]
       =\bar{\by}_h.
    \label{eq:delta-sum}
\end{equation}
Moreover,
\begin{equation}
    \restr{\bar{\by}_h-\bdelta_h(J,i)}{(N_h\setminus A_h)-i}
      =\E_I\!\left[
          \restr{\one_{I-\sigma_{I,J}(i)}}{(N_h\setminus A_h)-i}
        \right]
      \in P(\scr{N}_h/\{i\}).
    \label{eq:exchange-residual}
\end{equation}
Indeed, every set in the expectation becomes independent after adding $i$.
It also excludes $i$: if $i\in I$, property (iii) of
\Cref{lem:exchange} removes it.

For $i\in N_h\setminus A_h$ with $y_i>0$, define
\begin{equation}
    \brho_h(i)
      :=\E_{J\sim\nu_h}
          [\bdelta_h(J,i)\mid i\in J].
    \label{eq:rho-vector}
\end{equation}
Then convexity and \eqref{eq:exchange-residual} imply
\begin{equation}
    \restr{\bar{\by}_h-\brho_h(i)}{(N_h\setminus A_h)-i}
      \in P(\scr{N}_h/\{i\}).
    \label{eq:rho-feasible}
\end{equation}
Set $\brho_h(i):=0$ when $y_i=0$.  From now on, regard both
$\brho_h(i)$ and $\bar{\by}_h$ as vectors in $\R^N$ by extending them by zero
outside $N_h\setminus A_h$.  For a coordinate $i$ with $y_i>0$, define
\begin{equation}
    \by^{(i)}
      :=\left(\by-\sum_{h\in H(i)}\brho_h(i)\right)_+,
    \label{eq:candidate}
\end{equation}
where the positive part is coordinatewise.  Positive mass on $i$ implies
that $i$ is not a loop in any residual component containing it, so
$A+i\in\scr{F}$.

For every $h\in H(i)$, the vector $\bar{\by}_h-\brho_h(i)$ is nonnegative,
and its restriction to $(N_h\setminus A_h)-i$ belongs to
$P(\scr{N}_h/\{i\})$.  Moreover,
$\by^{(i)}\le \by-\brho_h(i)$ on this local ground set, because all the other
terms subtracted in \eqref{eq:candidate} are nonnegative.  Downward
closedness therefore gives
\[
    \restr{\by^{(i)}}{(N_h\setminus A_h)-i}
       \in P(\scr{N}_h/\{i\}).
\]
For a component not containing $i$, the residual matroid is unchanged and
$\by^{(i)}\le \by$ preserves feasibility.  The vector also satisfies
$\by^{(i)}\le \by\le\bx$.  Finally, its $i$-coordinate is zero.  Conditional on
$i\in J$, the exchange map charges coordinate $i$ exactly when $i\in I$, and
hence
\[
    [\brho_h(i)]_i=\bar y_{h,i}=y_i
    \qquad(h\in H(i)).
\]
Thus $\by^{(i)}$ is feasible for the LP defining $R(A+i)$.

It follows that
\begin{equation}
\begin{split}
    \Delta_i(A)
      &\le \bv\cdot(\by-\by^{(i)})\\
      &\le \sum_{h\in H(i)}\bv\cdot\brho_h(i).
\end{split}
    \label{eq:individual-charge}
\end{equation}
For a fixed component $h$, conditioning and \eqref{eq:delta-sum} give
\begin{equation}
\begin{split}
    \sum_{i\in N_h\setminus A_h}\bar y_{h,i}\brho_h(i)
      &=\E_{J\sim\nu_h}
          \left[\sum_{i\in J}\bdelta_h(J,i)\right]\\
      &\le\bar{\by}_h
      \qquad\text{coordinatewise.}
\end{split}
    \label{eq:rho-sum}
\end{equation}
Multiply \eqref{eq:individual-charge} by $y_i$ and sum.  Coordinates with
$y_i=0$ contribute nothing, while all positive coordinates are feasible
extensions.  Since $y_i=\bar y_{h,i}$ for $i\in N_h\setminus A_h$,
\begin{equation}
\begin{split}
    \sum_{i\notin A:\,A+i\in\scr{F}} y_i\Delta_i(A)
      &\le \sum_{h\in H}
          \sum_{i\in N_h\setminus A_h}
             \bar y_{h,i}\,\bv\cdot\brho_h(i)\\
      &\le \sum_{h\in H}\bv\cdot\bar{\by}_h\\
      &=\sum_{h\in H}\sum_{i\in N_h\setminus A_h}v_iy_i\\
      &=\sum_{i\in N}d(i)v_iy_i\\
      &\le k\,\bv\cdot\by
       =kR(A).
\end{split}
\end{equation}
This proves \eqref{eq:charging}.
\end{proof}

\subsection{Random-Order Analysis}
\label{sec:analysis}

Let $I(t)$ be the set of accepted elements of \Cref{alg:prophet} strictly before time $t \in [0,1]$. 
More, denote by $I$ the final accepted set after every element has arrived,
where $v(I)= \sum_{i \in I} v_i$ is the value of $I$.

We first decompose $v(I)$ into utility and revenue contributions.
Recalling the definition of $\Delta_i(I(T_i))$ from \eqref{eq:residual-drop}, 
set $\Util = \sum_{i \in I} (v_i - \gamma(T_i)\Delta_i(I(T_i)))$ to be the utility
and $\Rev = \sum_{i \in I} \gamma(T_i)\Delta_i(I(T_i))$ to be the revenue.
Then,
\begin{equation} \label{eqn:utility_rev_decomp}
    v(I) =\Util+\Rev.
\end{equation}
For our analysis, it will be convenient to define
\begin{equation}
    r(t):=\E[R(I(t))],
\end{equation}
where $R(I(t))$ is the residual value of $I(t)$ as defined in \eqref{eq:residual-value}.

We next lower bound $\E[\Util]$. Recalling that $\scr{F}$ is the feasibility constraint, we first specify the following function
for $i \in N$, $A \subseteq N$ and $t \in [0,1]$ as follows:
\[
    u_i(A,t):=
    \begin{cases}
      x_i\pos{v_i-\gamma(t)\Delta_i(A)},
        &i\notin A\text{ and }A+i\in\scr{F},\\
      0,&\text{otherwise}.
    \end{cases}
\]
Given this definition, we first observe that we can remove the conditioning on $T_i = t$ in the following way:
\begin{proposition} \label{prop:conditioning}
For each $i \in N$ and $t \in [0,1]$,
\begin{equation}
    \E[u_i(I(t),t)\mid T_i=t]\ge \E[u_i(I(t),t)].
    \label{eq:conditioning}
\end{equation}

\end{proposition}
\begin{proof}
This is precisely the argument of \cite{Lee2018}, and so we just sketch the details.
We can couple two executions, one with $T_i =t$, and the other being the usual execution.
If $i$ is accepted in the latter execution when $T_i < t$, then $u_i(I(t),t) = 0$, and so
the sample path version of \eqref{eq:conditioning} holds.
In all the other cases (i.e., $T_i \ge t$, or $T_i < t$ and $i$ is rejected),
the two executions have the same accepted set strictly before time $t$, and so again the sample path version of \eqref{eq:conditioning} holds.  
\end{proof}

\begin{lemma}[Utility]
\label{lem:utility}
The expected utility of \Cref{alg:prophet} satisfies
\[
    \E[\Util]\ge
       \int_0^1\bigl(1-k\gamma(t)\bigr)r(t)\,dt.
\]
\end{lemma}

\begin{proof}
First observe that
\begin{align*}
\E[\Util] &= \sum_{i \in N } \int_{0}^{1} \E[u_i(I(t),t)\mid T_i=t] dt \\
&\ge
       \int_0^1\E\!\left[\sum_{i\in N}u_i(I(t),t)\right]dt.
\end{align*}
For the equality, conditional on $T_i=t$, the indicator $X_i$ is independent of
$I(t)$. If $I(t)+i$ is feasible, the algorithm accepts $i$ exactly when it is active
and its value meets the threshold, so its conditional expected utility contribution is
$x_i\pos{v_i-\gamma(t)\Delta_i(I(t))}$. If $I(t)+i$ is infeasible, its contribution
is zero. These are precisely the two cases in the definition of $u_i(I(t),t)$.
The inequality applies \Cref{prop:conditioning}.

Fix a realization of $I(t)$ and let $\by$ be optimal as in
\eqref{eq:residual-value} for $I(t)$. For $i\in I(t)$, the constraint $y_i=0$
is imposed explicitly in \eqref{eq:residual-relaxation}. Write
$I_h(t):=I(t)\cap N_h$. If $i\in N_h\setminus I_h(t)$ is a loop in the local
contraction $\scr M_h/I_h(t)$, feasibility in its independent-set polytope gives
\[
    0\le y_i\le r_{\scr M_h/I_h(t)}(\{i\})=0.
\]
Thus this coordinate also has $y_i=0$. Since $I(t)$ is feasible, an element
$i\notin I(t)$ fails to extend $I(t)$ feasibly exactly when it is a loop in at
least one of these local contractions. Consequently, only coordinates with
$i\notin I(t)$ and $I(t)+i\in\scr F$ can have positive $y_i$.
For every such coordinate, $y_i\le x_i$, so
\begin{equation} \label{eq:monotone_coordinates}
  x_i\pos{v_i-\gamma(t)\Delta_i(I(t))}
    \ge y_i\bigl(v_i-\gamma(t)\Delta_i(I(t))\bigr).
\end{equation}
After summing and using the definition of $u_i(I(t),t)$, we can apply \eqref{eq:monotone_coordinates} followed by \Cref{lem:charging} to get
\[
   \sum_{i \in N} u_i(I(t),t) 
      \ge \bv\cdot\by-
          \gamma(t)\sum_{i\notin I(t):\,I(t)+i\in\scr{F}} y_i\Delta_i(I(t))
      \ge \bigl(1-k\gamma(t)\bigr)R(I(t)).
\]
Finally, taking expectations and integrating over $t$ proves the claim.
\end{proof}
We next compute the expected revenue:
\begin{lemma}[Revenue]
\label{lem:revenue}
\[
    \E[\Rev]
       =\gamma(0)R(\varnothing)+
          \int_0^1\gamma'(t)r(t)\,dt.
\]
\end{lemma}

\begin{proof}
If elements are accepted at times $t_1<\cdots<t_q$, denote the corresponding
residual drops by $D_1,\ldots,D_q$ (see \eqref{eq:residual-drop}). In this
case, we can write $R(I(t))=R(\varnothing)-\sum_{j:t_j<t}D_j$.
To spell out integration by parts for this step function, substitute this expression
and integrate each of its finitely many terms:
\begin{align*}
    \int_0^1\gamma'(t)R(I(t))\,dt
       &=R(\varnothing)\bigl(\gamma(1)-\gamma(0)\bigr)
         -\sum_{j=1}^qD_j\int_{t_j}^1\gamma'(t)\,dt\\
       &=-\gamma(0)R(\varnothing)+\sum_{j=1}^q\gamma(t_j)D_j.
\end{align*}
The values at the jump times do not affect the integral, and the last equality
uses $\gamma(1)=0$. Rearranging gives
\[
    \sum_{j=1}^q\gamma(t_j)D_j
      =\gamma(0)R(\varnothing)+
        \int_0^1\gamma'(t)R(I(t))\,dt.
\]
The left side is the realized revenue, and so the result follows after taking expectations.
\end{proof}

\begin{theorem}[Random-order ex-ante prophet inequality]
\label{thm:random-prophet}
For every $\bv\in\mb{R}_{\ge0}^N$, the accepted elements $I \subseteq N$ of \Cref{alg:prophet} satisfy
\[
    \E[v(I)]\ge
       \frac{1-e^{-k}}{k}\sum_{i\in N}x_iv_i.
\]
\end{theorem}

\begin{proof}
By taking expectations over \eqref{eqn:utility_rev_decomp} and applying \Cref{lem:utility,lem:revenue},
\[
 \E[v(I)]\ge \gamma(0)R(\varnothing)+
   \int_0^1\bigl(1-k\gamma(t)+\gamma'(t)\bigr)r(t)\,dt.
\]
The integral vanishes by \eqref{eq:gamma-ode}.  At $A=\varnothing$, the vector
$\bx$ is feasible in \eqref{eq:residual-value}, and the constraint
$\by\le\bx$ makes it optimal.  Hence
$R(\varnothing)=\sum_i x_iv_i$, while
$\gamma(0)=(1-e^{-k})/k$.
\end{proof}


\section{Conclusion}

We obtain ex-ante prophet inequalities and contention resolution schemes for
$k$-matchoids in both fixed adversarial order and random order.  The respective
ex-ante competitive ratios and selectability guarantees are $\frac{1}{k+1}$ and $\frac{1-e^{-k}}{k}$.

For fixed order, the algorithm coordinates a weighted principal partition in
every local matroid through one common surplus vector.  The resulting local
prices add to a single threshold for each element, and their aggregate budget
is at most $\sum_i d(i)p_i\le k p(N)$.  The construction extends to random-element
inputs with mutually exclusive activations within independent batches; this extension is given in
\Cref{sec:random-element-ocrs}.

For random order, the algorithm prices an element by its effect on a residual
LP over the local contracted matroid polytopes.  Fractional exchanges are
constructed separately in the local matroids and then combined.  Each
coordinate participates in at most $k$ local exchanges, which yields the
ex-ante competitive ratio $\frac{1-e^{-k}}{k}$ through the Lee--Singla utility--revenue analysis.

\section{AI disclosure}

The proof of the adversarial order result for the $k$-intersection of matroids was discovered by ChatGPT 5.6 Pro after being encouraged to try using \cite{KP2026} to improve the best known results in this setting. The authors verified the correctness and originality of all content including references.

\bibliographystyle{amsalpha}
\bibliography{ref}

\appendix

\section{Extended Preliminaries} \label{sec:extended_prelim}

\subsection{OCRSs and RCRSs}
\label{sec:ocrs}

Fix a downward-closed feasible family $\scr F\subseteq2^N$ and an activation
vector $\bx\in[0,1]^N$.  Let $R(\bx)$ contain each $i\in N$ independently
with probability $x_i$.  We call the elements of $R(\bx)$ \emph{active}.

An online contention resolution scheme (OCRS) observes, one element at a
time, whether the current element is active.  It immediately and irrevocably
decides whether to accept each active element and always maintains an accepted
set $I\in\scr F$.  Thus $I\subseteq R(\bx)$.  In the fixed-order model used
here, the order is chosen before the activations and algorithmic randomness
are drawn and is known to the scheme \cite{feldman2015online,Lee2018}.

\begin{definition}[Selectability]
An OCRS is $\alpha$-selectable on input $\bx$ if
\[
    \bP[i\in I]\ge\alpha x_i
    \qquad\text{for every }i\in N.
\]
\end{definition}

A random-order contention resolution scheme (RCRS) receives the elements in
a uniformly random order independent of $R(\bx)$, and its selection
probabilities also average over this order \cite{Lee2018,adamczyk2018random}.
Equivalently, one can draw independent clocks
$T_i\sim\operatorname{Unif}[0,1]$, process the elements in increasing order
of $T_i$, and obtain guarantees by averaging over $(T_i)_{i \in N}$.
\subsection{Reductions from Lee and Singla \cite{Lee2018}}
\label{sec:lee-singla-reductions}

We reproduce Theorem~1.4 and Lemma~3.3 from Lee and Singla \cite{Lee2018} for convenience.

\begin{theorem}[Lee--Singla~\cite{Lee2018}]
\label{thm:lee-singla-ocrs}
For $0\leq \alpha\leq 1$, an $\alpha$-ex-ante prophet inequality over $\scr{F}$ in adversarial order (respectively, random order) implies
an $\alpha$-selectable OCRS (resp., RCRS) over $\scr F$.
\end{theorem}

\begin{lemma}[Lee--Singla~\cite{Lee2018}]
\label{lem:bernoulli-reduction}
If there exists an $\alpha$-ex-ante prophet inequality for
weighted Bernoulli random variables, then there exists an
$\alpha$-ex-ante prophet inequality for generally distributed
independent random values. The same holds for random-order arrivals.
\end{lemma}
\begin{remark}
    The lemma in \cite{Lee2018} is proven for weighted Bernoulli random variables of the form
    $W_i = v_i \cdot \Ber(p_i)$ for each $i \in N$, where $v_i \ge 0$. They do not explicitly state that
    one can assume $(p_i)_{i \in N}$ lies in the polytope relaxation of the feasibility constraint. However, for such an input, the optimal solution to the ex-ante benchmark \eqref{eq:ex-ante-benchmark}, say $\bx = (x_i)_{i \in N}$, can be chosen to satisfy $x_i \le p_i$
    for all $i \in N$. Thus, one can always downsample by $x_i/p_i$ when processing $i$ with $p_i>0$, 
    and recover an input with weighted Bernoulli random variables of the form $v_i \cdot \Ber(x_i)$ for each $i \in N$.
\end{remark}

The policy-LP proof of \Cref{thm:lee-singla-ocrs} is pointwise in the
activation vector.  We use the following form from
\cite[Section~2.1]{Lee2018}.

\begin{corollary}
\label{cor:weighted-reduction}
Fix $\bx\in[0,1]^N$ and an arrival model.  Suppose that, for every
$\bv\in\R^N_{\ge0}$, there is a feasible online policy whose accepted elements $I$ satisfy
\[
    \E\!\left[\sum_{i\in I}v_i\right]
       \ge\alpha\sum_{i\in N}x_iv_i.
\]
Then there exists an $\alpha$-selectable contention resolution scheme on input
$\bx$ in the same arrival model.
\end{corollary}

These reductions are algorithmic under the usual oracle assumptions.  If a
weighted Bernoulli algorithm returns a policy in polynomial time and its
selection probabilities can be evaluated in polynomial time, the ellipsoid
argument of Lee and Singla computes, for every $\varepsilon>0$, a
polynomial-size distribution over deterministic online policies whose mixture
is $(\alpha-\varepsilon)$-selectable on input $\bx$
\cite[Theorem~1.4 and Section~2.2]{Lee2018}.  Their upper-quantile reduction is
efficient when the ex-ante program, the relevant quantiles, and the
conditional tail means can be computed efficiently
\cite[Lemma~3.3]{Lee2018}.  The weighted policies in this paper are
polynomial-time in the standard independence-oracle model: the convex price
programs use maximum-weight base and subgradient oracles, and the residual LP
in \eqref{eq:residual-value} has a separation oracle obtained by testing the
local contracted matroid polytopes \cite{Edmonds1970,GLS1981}.  The existential
reductions preserve the stated ex-ante competitive ratios and selectability guarantees exactly. Their polynomial-time
implementations incur an arbitrarily small additive loss.

\section{Adversarial-Order Random-Element OCRS}
\label{sec:random-element-ocrs}

This appendix records that the quadratic-price argument extends simultaneously to
local matroid constraints and to the random-element OCRS model of
\citet{MaMZ26Network}, in which at most one element is active in each batch.
We use $\bx$ for the activation marginals, reserving $\bp$ for the common-surplus
vector as in the adversarial-order proof.

Fix a $k$-matchoid with local matroids $(\scr{M}_h)_{h\in H}$ on $N$, and partition
the ground set into batches
\begin{equation}
    N=B_1\sqcup\cdots\sqcup B_T.
    \label{eq:batch-partition}
\end{equation}
Let $\bx\in\scr{P}_{\cap}$ satisfy $x(B_b)\le 1$ for every $b$.
For each batch $B_b$, at most one element is active, with
\begin{align}
    \bP\bigl[X_i=1,\ X_j=0\text{ for all }j\in B_b\setminus\{i\}\bigr]
        &=x_i \qquad (i\in B_b),
        \label{eq:random-element-marginals}\\
    \bP\bigl[X_j=0\text{ for all }j\in B_b\bigr]
        &=1-x(B_b).
        \label{eq:random-element-none}
\end{align}
The indicator vectors belonging to distinct batches are mutually independent.
The batches are revealed in the fixed order $B_1,\ldots,B_T$; when $B_b$
arrives, all of its activation indicators are revealed simultaneously.  A
random-element OCRS accepts only active elements, decides on each batch before
seeing any later batch, and always maintains a set in $\scr F$.  We use
$\alpha$-selectable in the sense of \Cref{sec:ocrs}, under the activation law
above.

\begin{theorem}[Random-element OCRS for matchoids]
\label{thm:random-element-matchoid}
For every $k$-matchoid, every batch partition \eqref{eq:batch-partition},
every $\bx\in\scr{P}_{\cap}$ with $x(B_b)\le1$ for all $b$, and every
corresponding random-element input
\eqref{eq:random-element-marginals}--\eqref{eq:random-element-none}, there
exists an $\alpha_k$-selectable random-element OCRS on this input, where
\[
    \alpha_k=\frac{1}{k+1}.
\]
We prove this through the following weighted guarantee.  For every $\bv\in\R^N_{\ge0}$ there are fixed thresholds
$\btau\in\R^N_{\ge0}$ and fixed local strengthenings
$\widehat{\scr{M}}_h\preceq\scr{M}_h$ such that the following almost
non-adaptive rule is feasible.  When a batch is revealed, if $i$ is its unique
active element, accept $i$ exactly when
\[
    v_i>\tau_i,
    \qquad
    (I\cup\{i\})\cap N_h\in\cI(\widehat{\scr{M}}_h)
        \quad\text{for every }h\in H(i).
\]
For every fixed order of the batches, the resulting set $I$ satisfies
\begin{equation}
    \E[v(I)]\ge \frac{1}{k+1}\sum_{i\in N}x_iv_i.
    \label{eq:random-element-weighted-guarantee}
\end{equation}
\end{theorem}

\begin{proof}
Fix $\bv\ge0$.  We first prove the weighted guarantee.  Delete the elements
with $x_i=0$, restrict every local matroid and every batch accordingly, and
restore the deleted elements as ineligible loops at the end.  Thus $x_i>0$
for every remaining element.  Since $\restr{\bx}{N_h}\in P(\scr{M}_h)$,
every element of $N_h$ is then a nonloop of $\scr{M}_h$.

For each $h\in H$, introduce a local price vector
$\bt_h\in\R^{N_h}_{\ge0}$.  Its contribution to element $i$ is defined only
when $h\in H(i)$, and the aggregate price is
\begin{equation}
    \tau_i:=\sum_{h\in H(i)}t_{h,i},
    \qquad
    \btau:=(\tau_i:i\in N).
    \label{eq:local-aggregate-price}
\end{equation}
Consider the convex program
\begin{equation}
\begin{split}
    \Phi\bigl((\bt_h)_{h\in H}\bigr)
       :=\sum_{h\in H}\Gamma_{\scr{M}_h}(\bt_h)
       +\frac12\sum_{i\in N}x_i(v_i-\tau_i)^2
       \quad\text{over }\prod_{h\in H}\R^{N_h}_{\ge0}.
\end{split}
    \label{eq:local-joint-program}
\end{equation}
As every local matroid is loopless on the remaining ground set,
$\Gamma_{\scr{M}_h}(\bt_h)\ge\frac12\|\bt_h\|_\infty^2$.
Hence the objective is coercive and has a minimizer.  Fix one and define
\begin{equation}
    \bp:=\bx\odot(\bv-\btau).
    \label{eq:local-surplus}
\end{equation}
The constrained first-order condition in the block $\bt_h$ is
\[
    0\in
    \partial\Gamma_{\scr{M}_h}(\bt_h)
       -\restr{\bp}{N_h}
       +N_{\R^{N_h}_{\ge0}}(\bt_h).
\]
Consequently, there is a vector
$\bg_h\in\partial\Gamma_{\scr{M}_h}(\bt_h)$ such that, for every
$i\in N_h$,
\begin{equation}
    g_{h,i}=p_i\quad\text{if }t_{h,i}>0,
    \qquad
    g_{h,i}\ge p_i\quad\text{if }t_{h,i}=0.
    \label{eq:local-boundary-optimality}
\end{equation}
By \eqref{eq:subdiff}, $\bg_h=\bt_h\odot\by_h$ for some
$\by_h\in B(\scr{M}_h)$, so $\bg_h\ge0$ and $g_{h,i}=0$ whenever
$t_{h,i}=0$.

Suppose that $p_i<0$.  The first relation in
\eqref{eq:local-boundary-optimality} rules out $t_{h,i}>0$ for every
$h\in H(i)$.  Thus $\tau_i=0$ and $p_i=x_iv_i\ge0$, a contradiction.
Therefore $\bp\ge0$.  If $t_{h,i}=0$, the second relation in
\eqref{eq:local-boundary-optimality} now gives
$0=g_{h,i}\ge p_i\ge0$.  Hence
\begin{equation}
    \restr{\bp}{N_h}\in
       \partial\Gamma_{\scr{M}_h}(\bt_h)
       \qquad(h\in H).
    \label{eq:local-common-subgradient}
\end{equation}
In particular, $\btau\le\bv$ and
\begin{equation}
    x_iv_i=p_i+x_i\tau_i
    \qquad(i\in N).
    \label{eq:local-exact-decomposition}
\end{equation}

Apply \Cref{lem:principal-blocks} separately to
$(\scr{M}_h,\bt_h,\restr{\bp}{N_h})$.  Let
$U_{h,1},\ldots,U_{h,\ell_h}$ be the price-level blocks, with prices
$\alpha_{h,1}>\cdots>\alpha_{h,\ell_h}\ge0$, and put
\[
    A_{h,0}:=\varnothing,
    \qquad
    A_{h,j}:=U_{h,1}\cup\cdots\cup U_{h,j}.
\]
Let
\[
    \scr{N}_{h,j}
       :=(\scr{M}_h/A_{h,j-1})|_{U_{h,j}},
    \qquad
    \widehat{\scr{M}}_h
       :=\scr{N}_{h,1}\oplus\cdots\oplus\scr{N}_{h,\ell_h}.
\]
We obtain
\begin{align}
    \widehat{\scr{M}}_h&\preceq\scr{M}_h,
        \label{eq:local-strengthening}\\
    p(S)&\le\alpha_{h,j} r_{\scr{N}_{h,j}}(S)
        \qquad(S\subseteq U_{h,j}),
        \label{eq:local-block-domination}\\
    \restr{\bx}{N_h}\cdot\bt_h&\le p(N_h).
        \label{eq:local-price-budget}
\end{align}

Put
\[
    P:=p(N),
    \qquad
    Q:=\{i\in N:p_i>0\}=\{i\in N:v_i>\tau_i\}.
\]
The online rule in the statement uses the strengthenings
$\widehat{\scr{M}}_h$ and accepts the unique active element of a batch exactly
when it belongs to $Q$ and is feasible in every relevant strengthening.
By \eqref{eq:local-strengthening}, the accepted set is in $\scr{F}$.

For $i\in Q\cap B_b$, call $i$ available when
\begin{equation}
    (I_{<b}\cup\{i\})\cap N_h
       \in\cI(\widehat{\scr{M}}_h)
       \qquad\text{for every }h\in H(i),
    \label{eq:prebatch-availability}
\end{equation}
where $I_{<b}$ is the set accepted before batch $B_b$ is revealed.
Availability is defined without regard to which element of $B_b$ is active.
Let
\[
    C:=\{i\in Q:i\text{ is unavailable at the start of its batch}\}.
\]
We claim that, for every realization of the batch outcomes,
\begin{equation}
    p(C)\le\tau(I).
    \label{eq:batched-blocked-surplus}
\end{equation}
Indeed, assign each $i\in C$ to one local matroid $h\in H(i)$ in which
\eqref{eq:prebatch-availability} fails, and then to the unique block
$U_{h,j}$ containing $i$.  Write $C_{h,j}$ for the set assigned to $(h,j)$.
The direct-sum definition gives
\[
    i\in\cl_{\scr{N}_{h,j}}(I_{<b}\cap U_{h,j})
    \qquad(i\in C_{h,j}).
\]
By monotonicity of closure,
$C_{h,j}\subseteq\cl_{\scr{N}_{h,j}}(I\cap U_{h,j})$.  Since
$I\cap U_{h,j}$ is independent in $\scr{N}_{h,j}$,
\[
    r_{\scr{N}_{h,j}}(C_{h,j})\le |I\cap U_{h,j}|.
\]
Using \eqref{eq:local-block-domination} and summing over all assigned pairs,
\begin{align*}
    p(C)
       &\le\sum_{h\in H}\sum_{j=1}^{\ell_h}
              \alpha_{h,j}|I\cap U_{h,j}|\\
       &=\sum_{h\in H}t_h(I\cap N_h)
        =\sum_{i\in I}\sum_{h\in H(i)}t_{h,i}
        =\tau(I),
\end{align*}
which proves \eqref{eq:batched-blocked-surplus}.

For $i\in Q$, let $A_i$ be the event that $i$ is available at the
start of its batch.  This event depends only on earlier batch outcomes and on
preprocessing randomness, and is therefore independent of the entire outcome
of the batch containing $i$.  Because at most one element in that batch is
active, the rule accepts $i$ exactly on
$A_i\cap\{X_i=1\}$.  Hence
\begin{equation}
    \bP[i\in I]=x_i\bP[A_i].
    \label{eq:batched-factorization}
\end{equation}
Using \eqref{eq:local-exact-decomposition},
\begin{align}
    \E[v(I)]
       &=\sum_{i\in Q}(p_i+x_i\tau_i)\bP[A_i]\notag\\
       &=P-\E[p(C)]+\E[\tau(I)] \ge P,
    \label{eq:batched-earns-surplus}
\end{align}
where the last inequality is \eqref{eq:batched-blocked-surplus}.

It remains to compare $P$ with the ex-ante value.  Summing the local price
budgets \eqref{eq:local-price-budget} and regrouping by elements gives
\begin{align}
    \bx\cdot\btau
       &=\sum_{h\in H}\restr{\bx}{N_h}\cdot\bt_h\notag\\
       &\le\sum_{h\in H}p(N_h)
        =\sum_{i\in N}d(i)p_i
        \le kP.
    \label{eq:matchoid-price-budget}
\end{align}
Together with \eqref{eq:local-exact-decomposition}, this yields
\[
    \sum_{i\in N}x_iv_i=P+\bx\cdot\btau\le(k+1)P.
\]
Combining this inequality with \eqref{eq:batched-earns-surplus} proves
\eqref{eq:random-element-weighted-guarantee}.

The weighted guarantee holds for every $\bv\ge0$.  The proof of
\Cref{cor:weighted-reduction} is finite-dimensional separation of achievable
acceptance vectors and does not use coordinatewise independence.  Applying it
to the class of online policies for the fixed random-element input gives an
$\alpha_k$-selectable random-element OCRS on this input.
\end{proof}

Taking singleton batches $B_b=\{i_b\}$ for any fixed arrival order
$i_1,\ldots,i_n$ makes the batch outcomes mutually independent Bernoulli
indicators with marginals $\bx$.  The weighted guarantee in
\Cref{thm:random-element-matchoid} then gives the Bernoulli case of
\Cref{thm:ex-ante-adversarial}.  Applying
\Cref{lem:bernoulli-reduction} gives the general ex-ante prophet inequality, and applying
\Cref{cor:weighted-reduction} gives \Cref{cor:ocrs}.

\subsection{Tightness}
\label{sec:random-element-tightness}

The random-element model makes the selectability guarantee in
\Cref{thm:random-element-matchoid} best possible when $k$ is a prime power. This is
the finite-affine-plane obstruction of \cite{MaMZ26Network}, translated into
$k$-matchoid language.

\begin{proposition}[Tightness for random-element OCRS]
\label{prop:random-element-tightness}
If $k \ge 2$ is a prime power, then no random-element OCRS for
$k$-matchoids is $\alpha$-selectable on every input for any
$\alpha>1/(k+1)$.  The same instances show that no online algorithm is better than
$1/(1+k)$-ex-ante competitive. 
\end{proposition}

\begin{proof}
The input is the finite-affine-plane
instance of \cite[Theorem~3.4, Definition~3.5, and Appendix~B.2 in the full
version]{MaMZ26Network}.
To translate its network-revenue-management
configuration into matchoid language, take the lines of the affine plane as
elements, its parallel classes as the batches, and, for every point
$u$, impose a rank-one local matroid on the lines incident to $u$.  Each
line lies in exactly $k$ local matroids.  Under this identification, both
the selectability bound and the weighted ex-ante bound are precisely as claimed.
\end{proof}

\end{document}